\documentclass{article}

\usepackage{amsmath}
\usepackage{amsfonts}
\usepackage{amssymb}
\usepackage{amsthm}
\usepackage{bm}
\usepackage{enumerate}
\usepackage{algorithm}
\usepackage{algorithmicx}
\usepackage{algpseudocode}
\usepackage{titlesec}
\usepackage[titletoc,toc,title]{appendix}

\usepackage[letterpaper, margin=1in]{geometry}

\allowdisplaybreaks

\newtheorem{theorem}{Theorem}[section]
\newtheorem{lemma}[theorem]{Lemma}

\newtheorem{remark}[theorem]{Remark}

\newcommand{\summ}{\displaystyle\sum}

\newcommand{\F}{\mathbb{F}}

\newcommand{\E}{\mathbb{E}}
\newcommand{\Var}{\mathrm{Var}}
\newcommand{\Cov}{\mathrm{Cov}}
\newcommand{\calC}{{\mathcal {C}}}
\def\polylog{\operatorname{polylog}}

\newcommand{\poly}{\mathsf{poly}}
\newcommand{\ba}{\mathbf{a}}
\newcommand{\deltaC}{\delta_{\mathcal C}}

\begin{document}
\title{Decoding Reed-Muller codes over product sets}
\author{John Kim\thanks{Department of Mathematics, Rutgers University. {\tt jonykim@math.rutgers.edu } Research supported by the National Science
Foundation Graduate Research Fellowship under Grant No.  DGE-1433187.
} \and
Swastik Kopparty\thanks{Department of Mathematics \& Department of Computer Science, Rutgers University. {\tt swastik@math.rutgers.edu }
Research supported in part by a Sloan Fellowship and NSF grant CCF-1253886.}}

\maketitle

\begin{abstract}
We give a polynomial time algorithm to decode multivariate polynomial codes of degree $d$ up to half their minimum distance, when the evaluation points are an arbitrary product set $S^m$, for every $d < |S|$. Previously known algorithms can achieve this only if the set $S$ has some very special algebraic structure, or if the degree $d$ is significantly smaller than $|S|$. We also give a near-linear time randomized algorithm, which is based on tools from list-decoding, to decode these codes from nearly half their minimum distance, provided $d < (1-\epsilon)|S|$ for constant $\epsilon > 0$.
 
Our result gives an $m$-dimensional generalization of the well
known decoding algorithms for Reed-Solomon codes, and can be viewed as giving an algorithmic
version of the Schwartz-Zippel lemma.
\end{abstract}

%
%
%

\section{Introduction}

Error-correcting codes based on polynomials have played an important role
throughout the history of coding theory. The mathematical phenomenon
underlying these codes is that distinct low-degree polynomials have different
evaluations at many points. More recently, the intimate relation between
polynomials and computation has led to polynomial-based error-correcting codes
having a big impact on complexity theory. Notable applications include
PCPs, interactive proofs, polynomial identity testing and property testing.

Our main result is a decoding algorithm for multivariate polynomial codes.
Let $\F$ be a field, let $S \subseteq \F$, let $d < |S|$ and let $m \geq 1$.
Consider the code of all $m$-variate polynomials of total
degree at most $d$, evaluated at all points of $S^m$:
$$ \mathcal C = \{   \langle P(\ba) \rangle_{\ba \in S^m} \mid P(X_1, \ldots, X_m) \in \F[X_1, \ldots, X_m], \deg(P) \leq d \}.$$
When $m = 1$, this code is known as the Reed-Solomon code~\cite{RS}, and for 
$m > 1$ this code is known as the Reed-Muller code~\cite{Muller,Reed}\footnote{The family of
Reed-Muller codes also includes polynomial evaluation codes where the individual degree $d$ is
larger than $|S|$, and the individual degree is capped to be at most $|S|-1$. We do not consider the 
$d \geq |S|$ case in this paper.}.

The code $\mathcal C$ above is a subset of $\F^{S^m}$, which we view as the space of functions
from $S^m$ to $\F_q$. Given two functions $f, g : S^m \to \F_q$,
we define their (relative Hamming) distance $\Delta(f, g) = \Pr_{\ba \in S^m} [ f(\ba) \neq g(\ba)].$
To understand the error-correcting properties of $\mathcal C$,
we recall the following well known lemma, often called the Schwartz-Zippel lemma:
\begin{lemma}
Let $\F$ be a field, and let $P(X_1, \ldots, X_m)$ be a nonzero polynomial over $\F$ with 
degree at most $d$. Then for every $S \subseteq \F$,
$$ \Pr_{\mathbf{a} \in S^m } [ P(\mathbf{a}) = 0 ] \leq \frac{d}{|S|} .$$
\end{lemma}
This lemma implies that for any two polynomials $P, Q$ of degree at most $d$,
$\Delta(P, Q) \geq (1 - \frac{d}{|S|})$. In other words
the minimum distance of $\mathcal C$ is at least $( 1 - \frac{d}{|S|})$.
It turns out that the minimum distance of $\mathcal C$ is in fact exactly $( 1 - \frac{d}{|S|})$,
and we let $\deltaC$ denote this quantity.

For error-correcting purposes, if we are given a ``received word" $r: S^m \to \F_q$
such that there exists a polynomial $P$ of degree at most $d$ with 
$\Delta(r, P) \leq \deltaC/2$, then we know that there is a unique such $P$.
The problem that we consider in this paper, ``decoding $\mathcal C$
upto half its minimum distance'', is the algorithmic task of finding this $P$.

\subsection{Our Results}

There is a rich history with several deep algebraic ideas surrounding the problem of
decoding multivariate polynomial codes. We first state our main results, and
then discuss its relationship to the various other known results.

\begin{theorem}[Efficient decoding of multivariate polynomial codes upto half their minimum distance]
\label{thm:intromain}
Let $\F$ be a finite field, let $S, d, m$ be as above, and let $\deltaC = (1 - \frac{d}{|S|})$.

There is an algorithm, which when given as input a function $r: S^m \to \F$,
runs in time $\poly(|S|^m , \log |\F|)$
finds the polynomial $P(X_1, \ldots, X_m)\in \F[X_1, \ldots, X_m]$ of degree at most $d$ (if any) such that:
$$\Delta(r, P) < \deltaC/2.$$
\end{theorem}

As we will discuss below, previously known efficient decoding algorithms for these codes only either worked
for (1) very algebraically special sets $S$, or (2) very low degrees $d$, or (3) decoded from a much smaller fraction
of errors ($\approx \frac{1}{m+1} \deltaC$ instead of $\frac{1}{2} \delta_C$). 

Using several further ideas, we also show how to implement the above algorithm in near-linear time to decode upto
almost half the minimum distance, provided $d$ is not $(1-o(1)) |S|$.
\begin{theorem}[Near-linear time decoding]
Let $\F$ be a finite field, let $S, d, m$ be as above, and let $\deltaC = (1 - \frac{d}{|S|})$. Assume
$\delta_C > 0$ is a constant.

There is a randomized algorithm, which when given as input a function $r: S^m \to \F$,
runs in time $|S|^m \cdot\poly( \log |S|^m, \log |\F|)$
finds the polynomial $P(X_1, \ldots, X_m)\in \F[X_1, \ldots, X_m]$ of degree at most $d$ (if any) with:
$$\Delta(r, P) < (1 - o(1)) \cdot \deltaC/2.$$
\end{theorem}

Over the rational numbers, we get a version of Theorem~\ref{thm:intromain} where the running time
is $\poly(|S|^m, t)$, where $t$ is the maximum bit-complexity of any point in $S$ or in the image of $r$.
This enables us to decode multivariate polynomial codes upto half the minimum distance
in the natural special case where the evaluation set $S$ equals $\{1, 2, \ldots, n \}$.

We also mention that decoding Reed-Muller codes over an arbitrary product set $S^m$ appears as a subroutine in the local decoding algorithm for multiplicity codes~\cite{KSY} (see Section 4 on ``Solving the noisy system"). Our results allow the local decoding algorithms there to run efficiently over all fields (\cite{KSY} could only do this over fields of small characteristic, where algebraically special sets $S$ are available). 


\subsection{Related work}

There have been many works studying the decoding of multivariate polynomial codes, which
prove (and improve) various special cases of our main theorem.

\paragraph{Reed-Solomon codes ($m=1$):} When $m = 1$, our problem is also known as
the problem of decoding Reed-Solomon codes upto half their minimum distance. That this
problem can be solved efficiently is very classical, and a number of 
algorithms are known for this (Mattson-Solomon~\cite{MSbook}, Berlekamp-Massey~\cite{Ma}, Berlekamp-Welch~\cite{WB}).
The underlying algorithmic ideas have subsequently had a tremendous impact on algebraic algorithms.

For Reed-Solomon codes, it is in fact known how to list-decode beyond half the minimum distance, upto the
Johnson bound (Guruswami-Sudan~\cite{GuSu}).
This has had numerous further applications in coding theory, complexity theory and pseudorandomness.

\paragraph{Special sets $S$:} For very special sets $S$, it turns out that there are some algebraic ways
to reduce the decoding of multivariate polynomial codes over $S^m$ to the decoding of univariate polynomial
codes. This kind of reduction is possible when $S$ equals the whole field $\F$, or more generally when 
$S$ equals an affine subspace over the prime subfield of $\F$.

When $S = \F_q$, then $S^m = \F_q^m$ and $S^m$ can then be identified with the large field $\F_{q^m}$ in a natural
$\F_q$-linear way (this understanding of Reed-Muller codes was discovered by~\cite{KLP}). This converts the multivariate setting into univariate setting, identifies
the multivariate polynomial code as a subcode of the univariate polynomial code, and (somewhat miraculously),
the minimum distance of the univariate polynomial code equals the minimum distance of the multivariate polynomial code.
Thus the classical Reed-Solomon decoding algorithms can then be used, and this leads to
an algorithm for the multivariate setting decoding upto half the minimum distance. In fact,
Pellikaan-Wu~\cite{PW} observed that this connection allows one to decode multivariate polynomial codes beyond half the minimum distance too, provided $S$ is special in the above sense.

Another approach which works in the case of $S = \F_q$ is based on local decoding. Here we use the fact that $S^m = \F_q^m$ contains many lines (not just the axis-parallel ones), and then use the univariate decoding algorithms to decode on those lines from $(1- \frac{d}{q})/2$ fraction errors. This approach manages to decode multivariate polynomial codes with $S = \F_q$
from $(\frac{1}{2} - o(1))$ of the minimum distance.
Again, this approach does not work for general $S$, since a general $S^m$ usually contains only axis-parallel lines (while $\F_q^m$ has many more lines).

\paragraph{Low degree $d$:}
When the degree $d$ of the multivariate polynomial code is significantly smaller than $|S|$, then a number of other
list-decoding based methods come into play.

The powerful Reed-Muller list-decoding algorithm of Sudan~\cite{Sudan} and its multiplicity-based generalization, based on
$(m+1)$-variate interpolation and root-finding, can decode from $1 - (\frac{d}{|S|})^{\frac{1}{m+1}}$ fraction
errors. With small degree $d = o(|S|)$ and $m = O(1)$, this decoding radius equals $1 - o(1)$! However when $d$
is much larger (say $0.9 \cdot |S|$), then the fraction of errors decodable by this algorithm is around $\frac{1}{m+1} \cdot (1 - \frac{d}{|S|}) = \frac{1}{m+1} \cdot \delta_C$.

Another approach comes from the list-decoding of tensor codes~\cite{GoGuRa}. While the multivariate polynomial codes we are interested in are not tensor codes, they are subcodes of the code of polynomials with {\em individual degree} at most $d$.
Using the algorithm of~\cite{GoGuRa} for decoding tensor codes, we get an algorithm that can decode from a $1-o(1)$ fraction of errors when $d = o(|S|)$, but fails to approach a constant fraction of the minimum distance when $d$ approaches $|S|$.

In light of all the above, to the best of our knowledge, for multivariate polynomial codes with $d > 0.9 \cdot |S|$ (i.e., $\delta_C < 0.1$), and
$S$ generic, the largest fraction of errors which could be corrected efficiently was about $\frac{1}{m+1}\delta_C$. In particular, 
the correctable fraction of errors is a vanishing fraction of the minimum distance, as the number of variables $m$ grows.

We thus believe it is worthwhile to investigate this problem, not only because of its basic nature, but also because of the many different
powerful algebraic ideas that only give partial results towards it.

\subsection{Overview of the decoding algorithm}


We now give a brief overview of our decoding algorithms.
Let us first discuss the bivariate ($m=2$) case. 
Here we are given a received word $r:S^2\to \F$ such
that there exists a codeword $P(X,Y) \in \F[X,Y]$ of degree at most $d = (1-\delta_C)|S|$ with $\Delta(P, r) < \frac{\delta_C}{2}$.
Our goal is to find $P(X,Y)$.

First some high-level strategy. An important role in our algorithm is played by the following observation:
the restriction of a degree $\leq d$
bivariate polynomial $P(X,Y)$ to a vertical line (fixing $X = \alpha$) or a horizontal line (fixing $Y = \beta$)
gives a degree $\leq d$ univariate polynomial. Perhaps an even more important role is played
by the following disclaimer: {\em the previous observation does not characterize bivariate polynomials
of degree $d$!} The set of functions $f: S^2 \to \F$ for which the horizontal restrictions and vertical restrictions
are polynomials of degree $\leq d$ is the code of polynomials with {\em individual degree} at most $d$ (this is the tensor
Reed-Solomon code, with much smaller distance than the Reed-Muller code). For such a function $f$ to be
in the Reed-Muller code, the different univariate polynomials that appear as horizontal and vertical restrictions must be related
in some way. The crux of our algorithm is to exploit these relations.

It will also help to recap the standard algorithm to decode
tensor Reed-Solomon codes upto half their minimum distance (this
scheme actually works for general tensor codes). 
Suppose we are given a received word $r: S^2 \to \F$, and we
want to find a polynomial $P(X,Y)$ with individual degrees 
at most $d$ which is close to $r$. 
One then takes the rows of this new received word (after having corrected the columns) and decodes them to the nearest degree $\leq d$ polynomial.
The key point is to pass some ``soft information" from the column decodings to the row decodings; the columns which were
decoded from more errors are treated with lower confidence. This 
decodes the tensor Reed-Solomon code from $1/2$ the minimum distance
fraction errors. Several ingredients from this algorithm will
appear in our Reed-Muller decoding algorithm.

Now we return to the problem of decoding Reed-Muller codes.
Let us write $P(X,Y)$ as a single 
variable polynomial in $Y$ with coefficients in $\F[X]$: $P(X,Y) = \summ_{i=0}^{d}{P_i(X)Y^{d-i}}$, where $\deg(P_i) \leq i$.
For each $\alpha \in S$, consider the restricted univariate polynomial $P(\alpha, Y)$.
Since $\deg(P_0) = 0$, $P_0(\alpha)$ must be the same for each $\alpha$. Thus all the polynomials $\langle P(\alpha, Y) \rangle _{\alpha \in S}$
have the same coefficient for $Y^d$. Similarly, the coefficients of $Y^{d-i}$ in the polynomials
 $\langle P(\alpha, Y)\rangle_{\alpha \in S}$ fit a degree $i$ polynomial. 

 
As in the tensor Reed-Solomon case, our algorithm begins by decoding each
column $r(\alpha, \cdot)$ to the nearest degree $\leq d$ univariate polynomial.
Now, instead of trying to use these decoded column polynomials to recover $P(X,Y)$ in one shot,
we aim lower and just try to recover $P_0(X)$. The advantage is that $P_0(X)$ is only a degree $0$
polynomial, and is thus resilient to many more errors than a degree $d$ polynomial.
Armed with $P_0(X)$, we then proceed to find $P_1(X)$. The knowledge of $P_0(X)$ allows us to
decode the columns $r(\alpha, \cdot)$ to a slightly larger radius; in turn this improved radius allows us to recover
the degree $1$ polynomial $P_1(X)$. 
At the $i^{\rm{th}}$ stage, we have already recovered $P_0(X)$, $P_1(X)$, \ldots, $P_{i-1}(X)$.
Consider, for each $\alpha \in S$, the function $f_\alpha(Y) = r(\alpha, Y) - \sum_{j = 0}^{i-1} P_j(\alpha)Y^{d-j}$.
Our algorithm decodes $f_{\alpha}(Y)$ to the nearest degree $d-i$ polynomial: note that as $i$ increases,
we are decoding to a lower degree polynomial, and hence we are able to handle a larger fraction of errors.
Define $h(\alpha)$ to be the coefficient of $Y^{d-i}$ in the polynomial so obtained; this
``should" equal the evaluation of the degree $i$ polynomial $P_i(\alpha)$. So we next decode $h(\alpha)$ to the nearest degree $i$
polynomial (using the appropriate soft information), and it turns out that this decoded polynomial must
equal $P_i(X)$. By the time $i$ reaches $d$, we would have recovered $P_0(X), P_1(X), \ldots, P_d(X)$,
and hence all of $P(X,Y)$. Summarizing, the algorithm repeatedly decodes the columns $r(\alpha, \cdot)$,
and at each stage it uses the relationship between the different univariate polynomial $P(\alpha, Y)$
to: (1) learn a little bit more about the polynomial $P(X,Y)$, and  (2) increase the radius to which
we can decode $r(\alpha, \cdot)$ in the next stage. This completes the description of the algorithm
in the $m = 2$ case.

The case of general $m$ is very similar, with only a small augmentation needed. Decoding $m$-variate
polynomials turns out to reduce to decoding $m-1$-variate polynomials with soft information; thus
in order to make a sustainable recursive algorithm, we aim a little higher and instead solve
the more general problem of decoding
multivariate polynomial codes with uncertainties (where each coordinate of the received word has an associated
``confidence" level).

To implement the above algorithms in near-linear time, we use some tools from list-decoding.
The main bottleneck in the running time is the requirement of having to decode the same column $r(\alpha, \cdot)$
multiple times to larger and larger radii (to lower and lower degree polynomials). To save on these
decodings, we can instead list-decode $r(\alpha, \cdot)$ to a large radius using a near-linear time list-decoder
for Reed-Solomon codes; this reduces the number of required decodings of the same column from $d$ to $O(1)$ (provided
$d < (1 - \Omega(1)) |S|$). For the $m=2$ case this works fine, but for $m > 2$ case this faces a serious obstacle;
in general it is impossible to efficiently list-decode Reed-Solomon codes {\em with uncertainties} beyond half the minimum distance of the code
(the list size can be superpolynomial). We get around this using some technical ideas,
based on speeding-up the decoding of Reed-Muller codes with uncertainties when the fraction of errors is significantly
smaller than half the minimum distance. For details, see Section~\ref{NearLinearRM}.

\subsection{Organization of this paper}

In Section 2, we cover the notion of weighted distance, which will be used in handling Reed-Solomon and Reed-Muller 
decoding with soft information on the reliability of the symbols in the encoding.  In Section 3, we state and prove 
a polynomial time algorithm for decoding bivariate Reed-Muller codes to half the minimum distance.  We then generalize 
the proof to decode multivariate Reed-Muller codes in Section 4.  Finally, in sections 5 and 6, we show that decoding Reed-Muller 
codes to almost half the minimum distance can be done in near-linear time by improving on the algorithms in Section 3 and 4.

\section{Preliminaries}

At various stages of the decoding algorithm,
we will need to deal with symbols and received words in which we
have varying amounts of confidence.
We now introduce some language to deal with such notions.

Let $\Sigma$ denote an alphabet.  A {\em weighted symbol} of $\Sigma$ is simply an element of $\Sigma \times [0,1]$.
In the weighted symbol $(\sigma, u)$, we will be thinking of $u \in [0,1]$ as our uncertainty that $\sigma$ is the symbol we should be talking about.

For a weighted symbol $(\sigma, u)$ and a symbol $\sigma'$,
we define their distance $\Delta( (\sigma, u), \sigma')$
by:
\begin{align*}
\Delta( (\sigma, u) , \sigma' )   =  \begin{cases} 1 - u/2   & \sigma \neq \sigma' \\ u/2  &  \sigma = \sigma' \end{cases}     
\end{align*}

For a weighted function $r : T \to \Sigma \times [0,1]$, and
a (conventional) function $f: T \to \Sigma$, we define their Hamming distance by
$$\Delta(r,f) = \sum_{t \in T} \Delta(r(t), f(t)).$$

The key inequality here is the triangle inequality.
\begin{lemma}[Triangle inequality for weighted functions]
Let $f, g: T \to \Sigma$ be functions, and let $r : T \to \Sigma \times [0,1]$ be
a weighted function.
Then:
$$ \Delta(r, f) + \Delta(r, g) \geq \Delta(f,g).$$
\end{lemma}
\begin{proof}
We will show that if $t \in T$ is such that
$f(t) \neq g(t)$, then $\Delta(r(t), f(t)) + \Delta(r(t), g(t)) \geq 1$.
This will clearly suffice to prove the lemma.

Let $r(t) = (\sigma, u)$. Suppose $f(t) = \sigma_1$ and $g(t) = \sigma_2$.
Then either $\sigma \neq \sigma_1$ or $\sigma \neq \sigma_2$, or both.
Thus either we have $\Delta(r(t), f(t)) + \Delta(r(t), g(t)) = (1-u/2) + u/2$
or we have $\Delta(r(t), f(t)) + \Delta(r(t), g(t)) = u/2 + (1-u/2)$,
or we have $\Delta(r(t), f(t)) + \Delta(r(t), g(t)) = (1-u/2) + (1-u/2)$.
In all cases, we have $\Delta(r(t), f(t)) + \Delta(r(t), g(t)) \geq 1$, as desired.
\end{proof}

The crucial property that this implies is the unique decodability up to
half the minimum distance of a code for {\em weighted} received words.

\begin{lemma}
Let $\calC \subseteq \Sigma^T$ be a code with minimum distance $\Delta$.
Let $r: T \to \Sigma \times [0,1]$ be a weighted function.
Then there is at most one $f \in \calC$ satisfying
$$\Delta(r, f) < \Delta/2.$$
\end{lemma}

\section{Bivariate Reed-Muller Decoding}

%

In this section, we provide an algorithm for decoding bivariate Reed-Muller codes to half the minimum distance.  
Consider the bivariate Reed-Muller decoding problem.  We are given a received word $r:S^2\to \F$.  Suppose that 
there is a codeword $C\in \F[X,Y]$ with $\deg(C)\leq d$, whose distance $\Delta(r,C)$ from the received word is at most 
half the minimum distance $n(n-d)/2$.  The following result says that there is a polynomial time algorithm in the size of 
the input $n^2$ to find $C$:

\begin{theorem}
\label{Bslow}
Let $\F$ be a finite field and let $S\subseteq \F$ be a nonempty subset of size $|S|=n$.  Given a received word 
$r:S^2\to\F$, there is a $O(n^3 \polylog (n,|\F|))$ time algorithm to find 
the unique polynomial (if it exists) $C\in\F[X,Y]$ with $\deg(C)\leq d$ such that 
$$\Delta(r,C) < \frac{n^2}{2}\left(1-\frac{d}{n}\right).$$
\end{theorem}

\subsection{Outline of Algorithm}
The general idea of the algorithm is to write $C(X,Y) = \summ_{i=0}^d{P_i(X) Y^{d-i}}\in \F[X][Y]$ as a 
polynomial in $Y$ with coefficients as polynomials in $\F[X]$, and attempt to uncover the coefficients 
$P_i(X)$ one at a time.

We outline the first iteration of the algorithm, which uncovers the coefficient $P_0(X)$ of degree $0$.
View the encoded message as a matrix on $S\times S$, where the rows are indexed by $x\in S$ 
and the columns by $y\in S$.  We first Reed-Solomon decode the rows $r(x,Y), x\in S$ to half the minimum distance 
$(n-d)/2$ and extract the coefficient of $Y^d$ in those decodings.  This gives us guesses for what $P_0(x)$ is for $x\in S$.  
However, this isn't quite enough to determine $P_0(X)$.  So we will also include some soft information 
which tells us how uncertain we are that the coefficient is correct.  The uncertainty is a number in $[0,1]$ that is 
based on how far the decoded codeword $G_x(Y)$ is from the received word $r(x,Y)$.  The farther apart, the higher the uncertainty.  
A natural choice for the uncertainty is simply the ratio of the distance $\Delta(G_x(Y),r(x,Y))$ to half the minimum distance 
$(n-d)/2$.  Let $f:S\to F\times[0,1]$ be the function of guesses for $P_0(x)$ and their uncertainties.  
We then use a Reed-Solomon decoder with uncertainties to find the degree $0$ polynomial that is closest to $f(X)$.  This will 
give us $P_0(X)$.  Finally, subtract $P_0(X) Y^d$ from $r(X,Y)$ and repeat to get the subsequent coefficients.

In the algorithm, we will use REED-SOLOMON-DECODER$(r,d)$ to denote the $O(n\polylog n)$ time algorithm that performs Reed-Solomon 
decoding of degree $d$ to half the minimum distance~\cite{Berl,WB}.  We use RS-SOFT-DECODER$(r,d)$ to denote the $O(n^2\polylog n)$ time algorithm 
that performs Reed-Solomon decoding of degree $d$ with uncertainties to half the minimum distance, which 
is based on Forney's generalized minimum distance decoding algorithm for concatenated codes~\cite{Forney}.  

\begin{algorithm}[H]
\caption{Decoding Bivariate Reed Muller\label{algBslow}}

\begin{algorithmic}[1]
\State Input: $r:S^{2}\to\F$.
\For{$i=0,1,\ldots,d$}
 \State Define $r_i:S\times S\to\F$ by 
 $$r_i(X,Y)=r(X,Y) - \summ_{j=0}^{i-1} Q_j(X)Y^{d-j}.$$
 \For{$x\in S$}
  \State Define $r_{i,x}:S\to\F$ by 
	$$r_{i,x}(Y)= r_i(x,Y).$$
	\State Define $G_{x}(Y)\in \F[Y]$ by
	$$G_{x}(Y)= \text{REED-SOLOMON-DECODER}(r_{i,x}(Y),d-i).$$
	\State $\sigma_{x}\gets \text{Coeff}_{Y^{d-i}}(G_{x})$.
	\State $\delta_{x}\gets \Delta(r_{i,x},G_{x})$.
	
 \EndFor
	\State Define the weighted function $f_i:S\to\F\times[0,1]$ by
	$$f_i(x)=\left(\sigma_{x},\frac{\delta_{x}}{(n-d+i)/2}\right).$$
 \State Define $Q_i:S\to \F$ by
	$$Q_i(X)= \text{RS-SOFT-DECODER}(f_i(X),i).$$
\EndFor
\State Output: $\summ_{i=0}^{d} Q_i(X)Y^{d-i}$.

\end{algorithmic}
\end{algorithm}

\subsection{Proof of Theorem~\ref{Bslow}}
\begin{proof}
\medskip\noindent\textbf{Correctness of Algorithm}
It suffices to show that $Q_i(X) = P_i(X)$ for $i=0,1,\ldots,d$, which we prove by induction.    
For this proof, the base case and inductive step can be handled by a single proof.  
We assume the inductive hypothesis that we have $Q_j(X)=P_j(X)$ for $j<i$.  
Note that the base case is $i=0$ and in this case, we assume nothing.  

It is enough to show $\Delta(f_i(X),P_i(X))<\frac{n}{2}\left(1-\frac{i}{n}\right)$.  Then 
$P_i(x)$ is the unique polynomial within weighted distance $\frac{n}{2}\left(1-\frac{i}{n}\right)$ 
of $f_i(X)$.  So RS-SOFT-DECODER$(f_i(X),i)$ will output $Q_i(X)=P_i(X)$.

We first show that $r_i(X,Y)$ is close to $C_i(X,Y)=\summ_{j=i}^{d} P_j(X)Y^{d-j}$.  
Observe that:

\begin{eqnarray*}
&&r_i(X,Y)-C_i(X,Y)\\
&=& (r_i(X,Y)+\summ_{j=1}^{i-1} P_j(X)Y^{d-j}) - (C_i(X,Y)+\summ_{j=1}^{i-1} P_j(X)Y^{d-j}))\\
&=& (r_i(X,Y)+\summ_{j=1}^{i-1} Q_j(X)Y^{d-j}) - C(X,Y)\\
&=& r(X,Y)-C(X,Y).
\end{eqnarray*}

Hence, 
$$\Delta(r_i(X,Y),C_i(X,Y))=\Delta(r(X,Y),C(X,Y))<\frac{n^2}{2}\left(1-\frac{d}{n}\right).$$

For each $x\in S$, define $C_{i,x}(Y) = C_i(x,Y)$.  
Define $\Delta_{x}=\Delta(r_{i,x}(Y),C_{i,x}(Y))$.  
Let $A = \{x\in S | G_{x}(Y) = C_{i,x}(Y)\}$ be the set of choices of $x$ 
such that $G_{x}(Y)=$ REED-SOLOMON-DECODER$(r_{i,x}(Y),d-i)$ produces $C_{i,x}(Y)$.  

Then, for $x\in A$, we have 
$$\delta_{x} = \Delta(r_{i,x}(Y),G_{x}(Y)) = \Delta(r_{i,x}(Y),C_{i,x}(Y)) = \Delta_{x}.$$

And for $x\notin A$, we have $G_{x}\neq C_{i,x}$, so
$$\delta_{x} = \Delta(r_{i,x}(Y),G_{x}(Y)) \geq n-d+i-\Delta(r_{i,x}(Y),C_{i,x}(Y)) = n-d+i-\Delta_{x}.$$

We now upper bound $\Delta(f_i(X),P_i(X))$:

\begin{eqnarray*}
\Delta(f_i(X),P_i(X)) &\leq & 
\summ_{x\in A}\frac{1}{2}\frac{\delta_{x}}{(n-d+i)/2}+\summ_{x\notin A}{1-\frac{1}{2}\frac{\delta_{x}}{(n-d+i)/2}}\\
&\leq & \summ_{x\in A}\frac{\Delta_{x}}{n-d+i}+\summ_{x\notin A}{1-\frac{n-d+i-\Delta_{x}}{n-d+i}}\\
&=& \summ_{x\in A}\frac{\Delta_{x}}{n-d+i}+\summ_{x\notin A}{\frac{\Delta_{x}}{n-d+i}}\\
&=& \summ_{x\in S^m}\frac{\Delta_{x}}{n-d+i}\\
&=& \frac{\Delta(r_i(X,Y),C_i(X,Y))}{n-d+i}\\
&<& \frac{n^2}{2}\left(1-\frac{d}{n}\right)\frac{1}{n-d+i}\\
&=& \frac{n}{2}\cdot\frac{n-d}{n-d+i}\\
&\leq& \frac{n}{2}\cdot\frac{n-i}{n}\\
&=& \frac{n}{2}\left(1-\frac{i}{n}\right).
\end{eqnarray*}

\medskip\noindent\textbf{Runtime of Algorithm}

We claim that the runtime of our algorithm is $O(n^3 \polylog n)$, 
ignoring the $\polylog |\F|$ factor from field operations.  The algorithm has $d+1$ 
iterations.  In each iteration, we update $r_i$, apply REED-SOLOMON-DECODER to $n$ rows and 
apply RS-SOFT-DECODER a single time to get the leading coefficient.  As updating takes 
$O(n^2)$ time, REED-SOLOMON-DECODER takes $O(n\polylog n)$ time, and RS-SOFT-DECODER takes $O(n^2\polylog n)$ time, 
we get $O(n^2\polylog n)$ for each iteration.  $d+1$ iterations gives a total runtime of $O(dn^2\polylog n)< O(n^3\polylog n)$.

\end{proof}

\section{Reed-Muller Decoding for General $m$}

We now generalize the algorithm for decoding bivariate Reed-Muller codes to handle Reed-Muller codes of 
any number of variables.  As before, we write the codeword as a polynomial in one of the variables and attempt to uncover 
its coefficients one at a time.  Interestingly, this leads us to a Reed-Muller decoding on one fewer variable, 
but with uncertainties.  This lends itself nicely to an inductive approach on the number of variables, however, 
the generalization requires us to be able to decode Reed-Muller codes with uncertainties.  This leads us to our main theorem:

\begin{theorem}
\label{RMD}
Let $\F$ be a finite field and let $S\subseteq \F$ be a nonempty subset of size $|S|=n$.  Given a received word 
with uncertainties $r:S^m\to\F\times[0,1]$, there is a $O(n^{m+2} \polylog (n,|\F|))$ time algorithm to find 
the unique polynomial (if it exists) $C\in\F[X_1,\ldots,X_m]$ with $\deg(C)\leq d$ such that 
$$\Delta(r,C) < \frac{n^m}{2}\left(1-\frac{d}{n}\right).$$
\end{theorem}

Note that to decode a Reed-Muller code without uncertainties, we may just set all the initial uncertainties to $0$.  
The algorithm slows by a factor of $n$ from the bivariate case due to having to use the RS-SOFT-DECODER instead of 
the faster REED-SOLOMON-DECODER on the rows of the received word.

\begin{proof}
The proof is by induction on the number of variables, and closely mirrors the proof of the bivariate case.

\medskip\noindent\textbf{Base Case}

We are given a received word with uncertainties $r:S\to\F\times[0,1]$ and asked to find the unique polynomial 
$C\in\F[X]$ with $\deg(C)\leq d$ within weighted distance $\frac{n-d}{2}$ of $r$.  This is just Reed-Solomon 
decoding with uncertainty, which can be done in time $O(n^2 \polylog n)$.

\medskip\noindent\textbf{Inductive Step}

Assume that the result holds for $m$ variables.  That is, assume we have access to an algorithm 
REED-MULLER-DECODER$(r,m,d)$ which takes as input a received word with uncertainties 
$r:S^m\to\F\times[0,1]$, and outputs the unique polynomial of degree at most $d$ (if it exists) within 
weighted distance $\frac{n^m}{2}\left(1-\frac{d}{n}\right)$ from $r$.  
We want to produce an algorithm for $m+1$ variables.  Before 
we progress, we set up some definitions to make the presentation and analysis of the algorithm cleaner.  
We are given $r:S^{m+1}\to\F\times[0,1]$.  View $r$ as a map from $S^m\times S\to\F\times[0,1]$, and write 
$r(\bm{X},Y) = (\overline{r}(\bm{X},Y),u(\bm{X},Y))$.  

Suppose that there exists a polynomial $C\in\F[\bm{X},Y]$ with $\deg(C)\leq d$ such that 
$$\Delta(r,C) < \frac{n^m}{2}\left(1-\frac{d}{n}\right).$$
View $C$ as a polynomial in $Y$ with coefficients in $\F[\bm{X}]$, $C(\bm{X},Y) = \summ_{i=0}^{d} P_i(\bm{X})Y^{d-i}$.  
The general strategy of the algorithm is to determine the $P_i$'s inductively by performing $d+1$ iterations from 
$i=0$ to $i=d$, and recovering $P_i(\bm{X})$ at the $i$-th iteration.

For the $i$-th iteration, consider the word 
$$r_i(\bm{X},Y) = \left(\overline{r}(\bm{X},Y) - \summ_{j=0}^{i-1} P_j(\bm{X})Y^{d-j},u(\bm{X},Y)\right).$$
Since $r$ is close to $\summ_{j=0}^{d} P_j(\bm{X})Y^{d-j}$, $r_i$ will be close to $C_i=\summ_{j=i}^{d} P_j(\bm{X})Y^{d-j}$.  
Our goal is to find $P_i(\bm{X})$, the leading coefficient of $C_i$ when viewed as a polynomial in $Y$.  
For each $\bm{x}\in S^m$, we decode the Reed-Solomon code with uncertainties given by $r_i(\bm{x},Y)$ and extract 
the coefficient of $Y^{d-i}$ along with how uncertain we are about the correctness of this coefficient.  This gives 
us a guess for the value $P_i(\bm{x})$ and our uncertainty for this guess.  We construct the function
$f_i:S^m\to F\times[0,1]$ of guesses for $P_i$ with their uncertainties.
We then apply the induction hypothesis of Theorem~\ref{RMD} to $f_i$ to recover $P_i$.

\begin{algorithm}[H]
\caption{Decoding Reed Muller with Uncertainties\label{alg}}

\begin{algorithmic}[1]
\State Input: $r:S^{m+1}\to\F\times[0,1]$.
\For{$i=0,1,\ldots,d$}
 \State Define $r_i:S^m\times S\to\F\times[0,1]$ by 
 $$r_i(\bm{X},Y)=\left(\overline{r}(\bm{X},Y) - \summ_{j=0}^{i-1} Q_j(\bm{X})Y^{d-j},u(\bm{X},Y)\right).$$
 \For{$\bm{x}\in S^m$}
  \State Define $r_{i,\bm{x}}:S\to\F\times[0,1]$ by 
	$$r_{i,\bm{x}}(Y)= r_i(\bm{x},Y).$$
	\State Define $G_{\bm{x}}(Y)\in \F[Y]$ by
	$$G_{\bm{x}}(Y)= \text{RS-SOFT-DECODER}(r_{i,\bm{x}}(Y),d-i).$$
	\State $\sigma_{\bm{x}}\gets \text{Coeff}_{Y^{d-i}}(G_{\bm{x}})$.
	\State $\delta_{\bm{x}}\gets \Delta(r_{i,\bm{x}},G_{\bm{x}})$.
	
 \EndFor
	\State Define the weighted function $f_i:S^m\to\F\times[0,1]$ by
	$$f_i(\bm{x})=\left(\sigma_{\bm{x}},\frac{\delta_{\bm{x}}}{(n-d+i)/2}\right).$$
 \State Define $Q_i:S^m\to \F$ by
	$$Q_i(\bm{X})= \text{REED-MULLER-DECODER}(f_i(\bm{X}),m,i).$$
\EndFor
\State Output: $\summ_{i=0}^{d} Q_i(\bm{X})Y^{d-i}$.

\end{algorithmic}
\end{algorithm}

\medskip\noindent\textbf{Correctness of Algorithm}

Suppose there is a polynomial $C(\bm{X},Y) = \summ_{i=0}^{d} P_i(\bm{X})Y^{d-i}$ such that 
$$\Delta(r(\bm{X},Y),C(\bm{X},Y)) < \frac{n^{m+1}}{2}\left(1-\frac{d}{n}\right).$$
We will show by induction that the $i$-th iteration of the algorithm produces $Q_i(\bm{X})=P_i(\bm{X})$.  
For this proof, the base case and inductive step can be handled by a single proof.  
We assume the inductive hypothesis that we have $Q_j(\bm{X})=P_j(\bm{X})$ for $j<i$.  
Note that the base case is $i=0$ and in this case, we assume nothing.

It is enough to show $\Delta(f_i(\bm{X}),P_i(\bm{X}))<\frac{n^m}{2}\left(1-\frac{i}{n}\right)$.  Then 
$P_i(\bm{X})$ is the unique polynomial within weighted distance $\frac{n^m}{2}\left(1-\frac{i}{n}\right)$ 
of $f_i(\bm{X})$.  So REED-MULLER-DECODER$(f_i(\bm{X}),m,i)$ will output $Q_i(\bm{X})=P_i(\bm{X})$.

We first show that $r_i(\bm{X},Y)$ is close to $C_i(\bm{X},Y)=\summ_{j=i}^{d} P_j(\bm{X})Y^{d-j}$.  
Observe that:

\begin{eqnarray*}
&&r_i(\bm{X},Y)-C_i(\bm{X},Y)\\
&=& (r_i(\bm{X},Y)+\summ_{j=1}^{i-1} P_j(\bm{X})Y^{d-j}) - (C_i(\bm{X},Y)+\summ_{j=1}^{i-1} P_j(\bm{X})Y^{d-j}))\\
&=& (r_i(\bm{X},Y)+\summ_{j=1}^{i-1} Q_j(\bm{X})Y^{d-j}) - C(\bm{X},Y)\\
&=& r(\bm{X},Y)-C(\bm{X},Y).
\end{eqnarray*}

Hence, 
$$\Delta(r_i(\bm{X},Y),C_i(\bm{X},Y))=\Delta(r(\bm{X},Y),C(\bm{X},Y))<\frac{n^{m+1}}{2}\left(1-\frac{d}{n}\right).$$

For each $\bm{x}\in S^m$, define $C_{i,\bm{x}}(Y) = C_i(\bm{x},Y)$.  
Define $\Delta_{\bm{x}}=\Delta(r_{i,\bm{x}}(Y),C_{i,\bm{x}}(Y))$.  
Let $A = \{\bm{x}\in S^m | G_{\bm{x}}(Y) = C_{i,\bm{x}}(Y)\}$ be the set of choices of $\bm{x}$ 
such that $G_{\bm{x}}(Y)=$ REED-SOLOMON-DECODER$(r_{i,\bm{x}}(Y),d-i)$ produces $C_{i,\bm{x}}(Y)$.  

Then, for $\bm{x}\in A$, we have 
$$\delta_{\bm{x}} = \Delta(r_{i,\bm{x}}(Y),G_{\bm{x}}(Y)) = \Delta(r_{i,\bm{x}}(Y),C_{i,\bm{x}}(Y)) = \Delta_{\bm{x}}.$$

And for $\bm{x}\notin A$, we have $G_{\bm{x}}\neq C_{i,\bm{x}}$, so
$$\delta_{\bm{x}} = \Delta(r_{i,\bm{x}}(Y),G_{\bm{x}}(Y)) \geq n-d+i-\Delta(r_{i,\bm{x}}(Y),C_{i,\bm{x}}(Y)) = n-d+i-\Delta_{\bm{x}}.$$

We now upper bound $\Delta(f_i(\bm{X}),P_i(\bm{X}))$:

\begin{eqnarray*}
\Delta(f_i(\bm{X}),P_i(\bm{X})) &\leq & 
\summ_{\bm{x}\in A}\frac{1}{2}\frac{\delta_{\bm{x}}}{(n-d+i)/2}+\summ_{\bm{x}\notin A}{1-\frac{1}{2}\frac{\delta_{\bm{x}}}{(n-d+i)/2}}\\
&\leq & \summ_{\bm{x}\in A}\frac{\Delta_{\bm{x}}}{n-d+i}+\summ_{\bm{x}\notin A}{1-\frac{n-d+i-\Delta_{\bm{x}}}{n-d+i}}\\
&=& \summ_{\bm{x}\in A}\frac{\Delta_{\bm{x}}}{n-d+i}+\summ_{\bm{x}\notin A}{\frac{\Delta_{\bm{x}}}{n-d+i}}\\
&=& \summ_{\bm{x}\in S^m}\frac{\Delta_{\bm{x}}}{n-d+i}\\
&=& \frac{\Delta(r_i(\bm{X},Y),C_i(\bm{X},Y))}{n-d+i}\\
&<& \frac{n^{m+1}}{2}\left(1-\frac{d}{n}\right)\frac{1}{n-d+i}\\
&=& \frac{n^m}{2}\cdot\frac{n-d}{n-d+i}\\
&\leq& \frac{n^m}{2}\cdot\frac{n-i}{n}\\
&=& \frac{n^m}{2}\left(1-\frac{i}{n}\right).
\end{eqnarray*}

\medskip\noindent\textbf{Runtime of Algorithm}

We claim the runtime of our $m$-variate Reed-Muller decoder is $O(n^{m+2} \polylog n)$, 
ignoring the $\polylog |\F|$ factor from field operations.  
We again proceed by induction on $m$.  In the base case of $m=1$, we simply run the 
Reed-Solomon decoder with uncertainties, which runs in $O(n^2 \polylog n)$ time.  
Now suppose the $m$-variate Reed-Muller decoder runs in time $O(n^{m+2} \polylog n)$.  
We need to show that the $m+1$-variate Reed-Muller decoder runs in time 
$O(n^{m+3} \polylog n)$.

The algorithm makes $d+1$ iterations.  In each iteration, we perform $n^m$ Reed-Solomon 
decodings with uncertainties, and extract the leading coefficient along with its uncertainty for each one.  
Each Reed-Solomon decoding takes $O(n^2\polylog n)$ time, while computing an uncertainty of a leading coefficient 
takes $O(n\polylog n)$.  So in this step, we have cumulative runtime $O(n^{m+2}\polylog n)$.  
Next we do a single $m$-variate Reed-Muller decoding with uncertainties, which takes $O(n^{m+2}\polylog n)$ 
by our induction hypothesis.  This makes the total runtime $O(dn^{m+2}\polylog n) \leq O(n^{m+3}\polylog n)$, 
as desired.

\end{proof}

\section{Near-Linear Time Decoding in the Bivariate Case}

In this section, we present our near-linear time, randomized decoding algorithm 
for bivariate Reed-Muller codes.



\begin{theorem}
\label{BRMD}
Let $\alpha \in (0,1)$ be a constant.
Let $\F$ be a finite field and let $S\subseteq \F$ be a nonempty subset of size $|S|=n$. Let
$d = \alpha n$.  Given a received 
word $r:S^2\to\F$, there is a $O(n^2 \polylog (n,|\F|))$ time, randomized algorithm to find 
the unique polynomial (if it exists) $C\in\F[X,Y]$ with $\deg(C)\leq d$ such that 
$$\Delta(r,C) < \frac{n^2}{2}\left(1-\alpha-\frac{1}{\sqrt{n}}\right).$$
\end{theorem}

\subsection{Outline of Improved Algorithm}

Recall that the decoding algorithms we presented in the previous sections
make $d+1$ iterations, where $d=\alpha n$, revealing 
a single coefficient of the nearest codeword during one iteration.  In a given 
iteration, we decode each row of $r_i(X,Y)$ to the nearest polynomial of degree 
$d-i$, extracting the coefficient of $Y^{d-i}$ and its uncertainty.  Then we 
Reed-Solomon decode with uncertainties to get the leading coefficient of $C(X,Y)$, 
when viewed as a polynomial in $Y$.

The runtime of this algorithm is $O(n^3 \polylog n)$.  Each iteration has $n$ 
Reed-Solomon decodings and a single Reed-Solomon decoding with uncertainties.  
As Reed-Solomon decoding takes $O(n \polylog n)$ time and Reed-Solomon decoding with 
uncertainties takes $O(n^2 \polylog n)$ time, we get a runtime of $O(n^3 \polylog n)$ 
with $d+1$ iterations.  To achieve near-linear time, we need to shave off a factor of $n$ on both the 
number of Reed-Solomon decodings and the runtime of Reed-Solomon decoding with 
uncertainties.  

To save on the number of Reed-Solomon decodings, we will 
instead list decode beyond half the minimum distance (using a near-linear time Reed-Solomon list-decoder),
and show that the list we  get is both small and essentially contains all of the decoded polynomials we require 
for $\Omega(n)$ iterations of $i$.  So we will do $O(n)$ Reed-Solomon list-decodings total 
instead of $O(n^2)$ Reed-Solomon unique decodings to half the minimum distance.

To save on the runtime of Reed-Solomon decoding with uncertainties, we will use
a probabilistic variant of Forney's generalized minimum distance decoding algorithm,
which runs in near-linear time, but reduces the decoding radius from $1/2$ the minimum
distance to $1/2 - o(1)$ of the minimum distance.


\subsection{Proof of Fast Bivariate Reed-Muller Decoding}

\begin{proof}[Proof of Theorem~\ref{BRMD}]

\medskip\noindent\textbf{Reducing the Number of Decodings}

To reduce the number of decodings, we will list decode past half the minimum distance.  
Let $r_{i,x}:S\to\F$ be a received word for a Reed-Solomon code $\calC_i$ of degree at most $d_i = d-i$.  
Let $t$ be the radius to which we list decode, and let $L_{i,x} = \{C\in \calC_i|\Delta(C,r_{i,x}) < t\}$ be 
the list of codewords within distance $t$ of $r_{i,x}$.  The radius to which we can decode while maintaining 
a polynomial-size list is given by the Johnson bound:
$$n(1-\sqrt{1-\delta_i}),$$

where $\delta_i = 1-\frac{d-i}{n} > 1-\frac{d}{n} = 1-\alpha$ is the relative distance of the code.  
By Taylor approximating the square root, we see that the Johnson bound exceeds half the minimum 
distance by $\Omega(n)$:

\begin{eqnarray*}
n(1-\sqrt{1-\delta_i}) &>& n(1-(1-\delta_i/2 + \delta_i^2/8 + 3\delta_i^3/16))\\
&=& n(\delta_i/2 + (1-\alpha)^2/8 + 3(1-\alpha)^3/16)\\
&=& (n-d+i)/2 + ((1-\alpha)^2/8)n + \epsilon n,
\end{eqnarray*}\

where $\epsilon=3(1-\alpha)^3/16$ is a positive constant.  By a standard list-size bound as in the one in Cassuto 
and Bruck~\cite{CB}, we see that if we set the list decoding radius $t = (n-d+i)/2 + ((1-\alpha)^2/8)n$, then the 
size of the list $|L_{i,x}|<\frac{1}{\epsilon}$ is constant.  So the list decoding radius exceeds half the minimum 
distance by $\Omega(n)$, and the list size is constant.  By Aleknovich's fast algorithm for weighted polynomial 
construction~\cite{Alek}, 
the list $L_{i,x}$ can be produced in time $(1/\alpha)^{O(1)}n \log^2 n \log\log n = O(n \polylog n)$.  
We will let RS-LIST-DECODER$(r,d,t)$ denote the Reed-Solomon list decoder that outputs a list of all 
ordered pairs of polynomials of degree at most $d$ within distance $t$ to the received word $r$ along 
with their distances to $r$.  Since the list size is constant, all of the distances can be computed in $O(n\polylog n)$ time.

\medskip\noindent\textbf{Faster Reed-Solomon Decoding with Uncertainties}

In the appendix, we give a description of the probabilistic GMD algorithm
that gives a faster Reed-Solomon decoder with uncertainties. 
We will refer to this algorithm as $\text{FAST-RS-DECODER}(f,i)$, where $f:S\to \F\times [0,1]$ 
is a received word with uncertainties, and $i$ is the degree of the code.  
$\text{FAST-RS-DECODER}(f,i)$ will output the codeword within distance $(n-i-\sqrt{n})/2$ 
(if it exists) with probability at least $1-\frac{1}{n^{\Omega(1)}}$ (the $\Omega(1)$
can be chosen to be an arbitrary constant, by simply repeating the algorithm independently several times).  
Therefore, in our final algorithm, with probability at least $99/100$, all invocations of the 
FAST-RS-DECODER will succeed.

\begin{algorithm}[H]
\caption{Decoding Bivariate Reed Muller\label{alg2}}

\begin{algorithmic}[1]
\State Input: $r:S^2\to\F$.
\State Let $c = ((1-\alpha)^2/8)$.
\For{$j=0,1,\ldots, \frac{d}{2cn}$}
 \State Let $t_j=\frac{n-d+j\cdot 2cn}{2} + cn$.
 \State Define $r_{j\cdot 2cn}:S\times S\to\F$ by 
 $$r_{j\cdot 2cn}(X,Y)=r(X,Y) - \summ_{i=0}^{j\cdot 2cn-1} Q_i(X)Y^{d-i}.$$
 \For{$x\in S$}
  \State Define $r_{j\cdot 2cn,x}:S\to\F$ by 
	$$r_{j\cdot 2cn,x}(Y)= r_{j\cdot 2cn}(x,Y).$$
	\State Define $\calC_{j\cdot 2cn}$ by
	$$\calC_{j\cdot 2cn} = \{C(Y)\in\F[Y]| \deg(C)<d-j\cdot 2cn\}.$$
	\State Define $L_{j,0,x} = \text{RS-LIST-DECODER}(r_{j\cdot 2cn,x}(Y),d-j\cdot 2cn,t_j)$.
 \EndFor
	
 \For{$k=0,1,\ldots,2cn-1$}
	\For{$x\in S$}
	 \State Define $(G_{x}(Y),\delta_x)\in L_{j,k,x}$ to be the unique codeword (if any) with 
	 $$\delta_x < \frac{n-d+j\cdot 2cn + k}{2}$$
	 \State $\sigma_{x}\gets \text{Coeff}_{Y^{d-j\cdot 2cn - k}}(G_{x})$.
	
	\EndFor
	 \State Define the weighted function $f_{j\cdot 2cn+k}:S\to\F\times[0,1]$ by
	 $$f_{j\cdot 2cn+k}(x)=\left(\sigma_{x},\frac{\delta_{x}}{(n-d+j\cdot 2cn + k)/2}\right).$$
	\State Define $Q_{j\cdot 2cn+k}:S\to \F$ by
	 $$Q_{j\cdot 2cn+k}(X)= \text{FAST-RS-DECODER}(f_{j\cdot 2cn+k}(X),j\cdot 2cn+k).$$
	
	\For{$x\in S$}
	 \State Define $L_{j,k+1,x} = \{(C-Q_{j\cdot 2cn+k}(x)Y^{d-j\cdot 2cn - k},\delta_{C,x})| C\in L_{j,k,x}, \text{Coeff}_{Y^{d-j\cdot 2cn - k}}(C)=Q_{j\cdot 2cn+k}(x)$\}.
	\EndFor
 \EndFor
\EndFor
\State Output: $\summ_{i=0}^{d} Q_i(\bm{X})Y^{d-i}$.

\end{algorithmic}
\end{algorithm}

\medskip\noindent\textbf{Correctness of Algorithm}

View the received word as a matrix on $S\times S$, where the rows are indexed by $x\in S$ 
and the columns by $y\in S$.  For correctness, we have to show two things.  
First, that Algorithm~\ref{alg2} produces the same 
row decodings $G_{x}(Y)$ as Algorithm~\ref{alg}.  Second, that the algorithm actually 
extracts the coefficients of $C(X,Y) = \summ_{i=0}^{d} P_i(X)Y^{d-i}$ when viewed as a polynomial in $Y$, 
i.e. $Q_i(X) = P_i(X)$ for $i=0,\ldots,d$.  
Define $r_{j\cdot 2cn + k}:S\times S\to\F$ by 
$$r_{j\cdot 2cn + k}(X,Y)=r(X,Y) - \summ_{i=0}^{j\cdot 2cn + k -1} Q_i(X)Y^{d-i},$$ 
and define $r_{j\cdot 2cn,x}:S\to\F$ by 
$$r_{j\cdot 2cn + k,x}(Y)= r_{j\cdot 2cn+k}(x,Y).$$

Then we want to show that in each of the $d+1$ iterations of $(j,k)$, we have
$$G_{x}(Y) = \text{REED-SOLOMON-DECODER}\left(r_{j\cdot 2cn + k,x}(Y),d-j\cdot 2cn - k\right).$$

It is enough to instead show that the list $L_{j,k,x}$ contains all the polynomials of degree at 
most $d-j\cdot 2cn-k$ within distance $t_j = (n-d+j\cdot 2cn)/2 + cn > (n-d+j\cdot 2cn + k)/2$ of $r_{j\cdot 2cn + k,x}(Y)$.
Furthermore, we want to show $Q_{j\cdot 2cn + k}(X) = P_{j\cdot 2cn + k}(X)$.

We prove this by induction on $(j,k)$.  The base case is $j=k=0$.  For each row $x\in S$, 
we have 
$$L_{0,0,x} = \text{RS-LIST-DECODER}(r_{j\cdot 2cn,x}(Y),d-j\cdot 2cn,t_0).$$

The induction hypothesis is that for every $(j',k') < (j,k)$ in the lexicographic order, we have 
$L_{j',k',x} = \{(\overline{C},\Delta(\overline{C},r_{j'\cdot 2cn + k',x})) | 
\overline{C}\in \calC_{j'\cdot 2cn + k'},\Delta(\overline{C},r_{j'\cdot 2cn + k',x}) < t_{j'}\}$ and that 
$Q_{j'\cdot 2cn + k'}(X) = P_{j'\cdot 2cn + k'}(X)$.  We will show the corresponding statements hold 
true for $(j,k)$.

If $k=0$, then the fact that the algorithm extracted the correct coefficients thus far means that the $r_{j\cdot 2cn}$ 
are the same in both Algorithm~\ref{alg} and Algorithm~\ref{alg2}.  Since 
$L_{j,0,x} = \text{RS-LIST-DECODER}(r_{j\cdot 2cn,x}(Y),d-j\cdot 2cn,t_j),$ 
the induction hypothesis on $L_{j,0,x}$ is met by the definition of RS-LIST-DECODER.  

If $k\neq 0$, then we know from the induction hypothesis that 
$L_{j,k-1,x} = \{(\overline{C},\Delta(\overline{C},r_{j\cdot 2cn + k-1,x})) | 
\overline{C}\in \calC_{j\cdot 2cn + k-1},\Delta(\overline{C},r_{j\cdot 2cn + k-1,x}) < t_j\}$.  
We want to say that 
$$L_{j,k,x} = \{(\overline{C},\Delta(\overline{C},r_{j\cdot 2cn + k,x})) | 
\overline{C}\in \calC_{j\cdot 2cn + k},\Delta(\overline{C},r_{j\cdot 2cn + k,x}) < t_j\}$$.

We defined $L_{j,k,x}$ in terms of $L_{j,k-1,x}$ to be:
$$\{(\overline{C}-Q_{j\cdot 2cn+k-1}(x)Y^{d-j\cdot 2cn - k+1},\Delta(\overline{C},r_{j\cdot 2cn + k-1,x}))| 
\overline{C}\in L_{j,k-1,x}, \text{Coeff}_{Y^{d-j\cdot 2cn - k+1}}(\overline{C})=Q_{j\cdot 2cn+k-1}(x)\}.$$

As $Q_{j\cdot 2cn+k-1}(X) = P_{j\cdot 2cn+k-1}(X)$, $L_{j,k,x}$ is essentially obtained by taking the 
codewords with the correct leading coefficients and subtracting off the leading term.  We claim that what we get 
is the set of all polynomials of degree at most $d-j\cdot 2cn-k$ within distance $t_j$ of $r_{j\cdot 2cn + k,x}$.  

Consider any $(G,\delta)\in L_{j,k,x}$.  By definition of $L_{j,k,x}$, we know there exists a $\overline{C} \in L_{j,k-1,x}$ 
with $\text{Coeff}_{Y^{d-j\cdot 2cn - k+1}}(\overline{C}) = Q_{j\cdot 2cn+k-1}(x)$ such that 
$$(G,\delta) = (\overline{C}-Q_{j\cdot 2cn+k-1}(x)Y^{d-j\cdot 2cn - k+1},\Delta(\overline{C},r_{j\cdot 2cn + k-1,x})).$$

So we have
\begin{eqnarray*}
\overline{C} &=& G+Q_{j\cdot 2cn+k-1}(x)Y^{d-j\cdot 2cn - k+1}\\
\delta &=& \Delta(\overline{C},r_{j\cdot 2cn + k-1,x}) < t_j.
\end{eqnarray*}

As $\text{Coeff}_{Y^{d-j\cdot 2cn - k+1}}(\overline{C}) = Q_{j\cdot 2cn+k-1}(x)$, we have $\deg(G)$ is at most $d-j\cdot 2cn - k$.  
Also, as $r_{j\cdot 2cn + k-1,x} = r_{j\cdot 2cn + k,x} + Q_{j\cdot 2cn+k-1}(x)Y^{d-j\cdot 2cn - k+1}$, we have 
$\Delta(G,r_{j\cdot 2cn+k,x})=\Delta(\overline{C},r_{j\cdot 2cn + k-1,x}) = \delta < t_j$.  
Hence, $G$ is a polynomial of degree at most $d-j\cdot 2cn-k$ within distance $t_j$ of $r_{j\cdot 2cn + k,x}$.  

For the reverse inclusion, suppose $G$ is a polynomial of degree at most $d-j\cdot 2cn-k$ at distance $\delta < t_j$ of $r_{j\cdot 2cn + k,x}$.  
Then $\overline{C} := G+Q_{j\cdot 2cn+k-1}(x)Y^{d-j\cdot 2cn - k+1} \in L_{j,k-1,x}$.  
Since $\text{Coeff}_{Y^{d-j\cdot 2cn - k+1}}(\overline{C}) = Q_{j\cdot 2cn+k-1}(x)$, we get that 
$G = \overline{C} - Q_{j\cdot 2cn+k-1}(x)Y^{d-j\cdot 2cn - k+1} \in L_{j,k,x}$, as desired.

It remains to show that $Q_{j\cdot 2cn+k}(X) = P_{j\cdot 2cn+k}(X)$.  As in the proof of Theorem~\ref{RMD}, we show 
that $\Delta(f_{j\cdot 2cn+k}(X), P_{j\cdot 2cn+k}(X)) < \frac{n-j-\sqrt{n}}{2}$, so that the output of $\text{FAST-RS-DECODER}(f_{j\cdot 2cn+k}(X),j)$ 
is $P_{j\cdot 2cn+k}(X)$.  Using the first part of the induction we just proved, we get the same $f_{j\cdot 2cn+k}(X)$ as in Algorithm~\ref{alg}.  
This means we can adopt a nearly identical argument to get to this step:
$$\Delta(f_{j\cdot 2cn+k}(X), P_{j\cdot 2cn+k}(X)) \leq \frac{\Delta(r_{j\cdot 2cn+k}(X,Y),C_{j\cdot 2cn+k}(X,Y))}{n-d+j\cdot 2cn+k}.$$

From here, we get:
\begin{eqnarray*}
\Delta(f_{j\cdot 2cn+k}(X), P_{j\cdot 2cn+k}(X)) &<& \frac{n^2}{2}\left(1-\frac{d}{n}-\frac{1}{\sqrt{n}}\right)\frac{1}{n-d+j\cdot 2cn+k}\\
&=& \frac{n}{2}\cdot\frac{n-d-\sqrt{n}}{n-d+j\cdot 2cn+k}\\
&\leq& \frac{n}{2}\cdot\frac{n-j\cdot 2cn-k-\sqrt{n}}{n}\\
&=& \frac{n-j\cdot 2cn-k-\sqrt{n}}{2}.
\end{eqnarray*}

\medskip\noindent\textbf{Analysis of Runtime of Bivariate Reed-Muller Decoder}

We run RS-LIST-DECODER $\frac{d}{2cn}n = \frac{\alpha}{2c}n = \frac{4\alpha}{(1-\alpha)^2}n$ times.  Also, we run 
FAST-RS-DECODER $d=\alpha n$ times.  As both of these algorithms run in $O(n \polylog n)$ time, the total runtime of the 
algorithm is $O(n^2 \polylog (n,|\F|))$, after accounting for field operations.  
As the input is of size $n^2$, this is near-linear in the size of the input.

\end{proof}

\section{Near-Linear Time Decoding in the General Case}
\label{NearLinearRM}

A more involved variation of the near-linear time, randomized decoding algorithm for 
bivariate Reed-Muller codes can be used to get a near-linear time, randomized algorithm for decoding 
Reed-Muller codes on any number of variables:

\begin{theorem}
\label{GenRMD}
Let $\F$ be a finite field and let $S\subseteq \F$ be a nonempty subset of size $|S|=n$.  
Let $\beta > \frac{1}{2}$.  
Given a received word  
$r:S^m\to\F$, there is a $O\left(n^m\cdot \polylog (n,|\F|)\right)$ time, randomized algorithm to find 
the unique polynomial (if it exists) $C\in\F[X_1,\ldots,X_m]$ with $\deg(C)\leq d$ such that 
$$\Delta(r,C) < \frac{n^m}{2}\left(1-\frac{d + (m-1)\beta\sqrt{n}}{n}\right).$$
\end{theorem}

As part of the algorithm for near linear time Reed-Muller decoding, we will need to decode 
Reed-Muller codes with uncertainties to various radii less than half their minimum distance.  
We require the following theorem to do such decodings efficiently.

\begin{theorem}
\label{UncRMD}
Let $\F$ be a finite field and let $S\subseteq \F$ be a nonempty subset of size $|S|=n$.  
Let $\beta > \frac{1}{2}$, and let $e$ be an integer satisfying $0\leq e < n-d-m\beta\sqrt{n}$.  
Given a received word with uncertainties 
$r:S^m\to\F\times [0,1]$, there is a $O\left(\frac{n^{m+1}}{e+1}\cdot \polylog (n,|\F|)\right)$ time algorithm to find 
the unique polynomial (if it exists) $C\in\F[X_1,\ldots,X_m]$ with $\deg(C)\leq d$ such that 
$$\Delta(r,C) < \frac{n^{m}}{2}\left(1-\frac{d + m\beta\sqrt{n} + e}{n}\right).$$
\end{theorem}

\begin{remark}
The algorithm requires the use of the FAST-RS-DECODER to a radius that is $\beta\sqrt{n}$ less than half the minimum 
distance.  As long as $\beta > \frac{1}{2}$, the FAST-RS-DECODER runs in $O(n\polylog n)$ time.
\end{remark}

\begin{proof}[Proof of Theorem~\ref{UncRMD}]
The proof is by induction on the number of variables $m$.  The proof of the base case of $m=2$ is similar to the 
proof of the inductive step and will be handled last.  
Assume the theorem statement is true for $m$, and let 
RM-UNC-DECODER$(f,d,s)$ denote the $O\left(\frac{n^{m+1}}{e+1}\cdot \polylog (n,|\F|)\right)$ 
time algorithm that finds the unique polynomial (if it exists) of degree at most $d$ within distance $s$ from $f$, where 
$f:S^m\to\F\times [0,1]$ and $s$ can be written as $\frac{n^m}{2}\left(1-\frac{d + m\beta\sqrt{n} + e}{n}\right)$.  
We want to show that the theorem statement holds for $m+1$ variables.

\begin{algorithm}[H]
\caption{Decoding Reed Muller with Uncertainties\label{alg4}}

\begin{algorithmic}[1]
\State Input: $r:S^{m+1}\to\F\times [0,1]$.
\For{$j=0,1,\ldots, \frac{d}{e+1}$}
 \State Let $t_j=\frac{n-d+j\cdot (e+1) - \beta\sqrt{n}}{2}$.
 \State Define $r_{j\cdot (e+1)}:S^m\times S\to\F$ by 
 $$r_{j\cdot (e+1)}(\bm{X},Y) = r(\bm{X},Y) - \summ_{i=0}^{j\cdot (e+1)-1} Q_i(\bm{X})Y^{d-i}.$$
 \For{$\bm{x}\in S^m$}
  \State Define $r_{j\cdot (e+1),\bm{x}}:S\to\F$ by 
	$$r_{j\cdot (e+1),\bm{x}}(Y)= r_{j\cdot (e+1)}(\bm{x},Y).$$

	\State Define $D_{j,0,\bm{x}}(Y) = \text{FAST-RS-DECODER}(r_{j\cdot (e+1),\bm{x}}(Y),d-j\cdot (e+1),t_j)$.
	\State Define $\delta_{\bm{x}} = \Delta(D_{j,0,\bm{x}}(Y),r_{j\cdot (e+1),\bm{x}}(Y))$.
 \EndFor
	
 \For{$k=0,1,\ldots,e$}
  \For{$\bm{x}\in S^m$} 
	 \If{$\deg(D_{j,k,\bm{x}}(Y)) \leq d-j\cdot (e+1) - k$}
	  $$\sigma_{\bm{x}} \gets \text{Coeff}_{Y^{d-j\cdot (e+1) - k}}(D_{j,k,\bm{x}}(Y)).$$
	 \EndIf
	\EndFor
  
	 \State Define the weighted function $f_{j\cdot (e+1)+k}:S^m\to\F\times[0,1]$ by
	 $$f_{j\cdot (e+1)+k}(\bm{x})=\left(\sigma_{\bm{x}},\min\left\{1,\frac{\delta_{\bm{x}}}{(n-d+j\cdot (e+1) + k - \beta\sqrt{n} - e)/2}\right\}\right).$$
	\State Define $Q_{j\cdot (e+1)+k}:S^m\to \F$ by
	 $$Q_{j\cdot (e+1)+k}(\bm{X})= \text{RM-UNC-DECODER}\left(f_{j\cdot (e+1)+k}(\bm{X}),j\cdot (e+1)+k, \frac{n^m}{2}\left(1-\frac{j\cdot (e+1)+k+m\beta\sqrt{n}}{n-d+j\cdot (e+1)+k}\right)\right).$$
	
	\For{$\bm{x}\in S^m$}
	 \State Define $D_{j,k+1,\bm{x}}:S\to\F$ by 
	  $$D_{j,k+1,\bm{x}} = D_{j,k,\bm{x}} - Q_{j\cdot (e+1)+k}(\bm{x}) Y^{d-j\cdot (e+1) - k}.$$
	\EndFor
 \EndFor
\EndFor
\State Output: $\summ_{i=0}^{d} Q_i(\bm{X})Y^{d-i}$.

\end{algorithmic}
\end{algorithm}

The algorithm proceeds as follows: 
As before, we write $C(\bm{X},Y) = \summ_{i=0}^{d} {P_i(\bm{X}) Y^{d-i}}$, and find the $P_i$ iteratively.  
In the $i$-th iteration, decode row $r_{i,\bm{x}}$, $\bm{x}\in S^m$ to a degree $d-i$ polynomial within radius 
$\frac{1}{2}(n-d+i-\beta \sqrt{n} -e)$ to get $D_{i,\bm{x}}(Y)$.  
To reduce the number of times we decode, we will instead decode to the larger radius $\frac{1}{2}(n-d+i-\beta \sqrt{n})$ 
and use this decoding for $e+1$ iterations.
Construct the function $f_i:S^m\to \F\times [0,1]$ of (leading coefficient, uncertainty) 
$= \left(\text{Coeff}_{Y^{d-i}}(D_{i,\bm{x}}),\frac{\Delta(r_{i,\bm{x}},D_{i,\bm{x}})}{(n-d+i-\beta \sqrt{n}-e)/2}\right)$.  
Finally, decode $f_i(\bm{X})$ to a degree $i$ polynomial within radius $\frac{n^m}{2}\left(1-\frac{i+m\beta\sqrt{n}}{n-d+i}\right)$ 
to get $Q_i(\bm{X})$.

\medskip\noindent\textbf{Proof of Correctness}

We have to show $Q_i(\bm{X}) = P_i(\bm{X})$.  It is enough to show that 
$$\Delta(f_i,P_i) < \frac{n^m}{2}\left(1-\frac{i+m\beta\sqrt{n}}{n-d+i}\right) < \frac{n^m}{2}\left(1-\frac{i}{n}\right).$$  
Then $P_i$ will be the unique polynomial of degree $i$ within distance $\frac{n^m}{2}\left(1-\frac{i+m\beta\sqrt{n}}{n-d+i}\right)$ 
of $f_i$.  Since $Q_i$ is a polynomial of degree $i$ within distance $\frac{n^m}{2}\left(1-\frac{i+m\beta\sqrt{n}}{n-d+i}\right)$ 
of $f_i$, $Q_i$ must be equal to $P_i$.

When we decode $r_{i,\bm{x}}$ to radius $\frac{1}{2}(n-d+i-\beta \sqrt{n}-e)$, there are four possibilities:

\begin{enumerate}
\item
The decoding is unsuccessful.  In this case, we set $D_{i,\bm{x}}$ to be any polynomial of degree $n-d+i$ and set 
the uncertainty $u_i = 1$.  The contribution to $\Delta(f_i,P_i)$ is $\Delta(f_i(\bm{x}),P_i(\bm{x})) = 1/2$, 
which is bounded above by $\frac{1}{2}\frac{\Delta(r_{i,\bm{x}},C_{i,\bm{x}})}{(n-d+i-\beta \sqrt{n}-e)/2}$.

\item
The decoding succeeds and is correct.  In this case, $D_{i,\bm{x}} = C_{i,\bm{x}}$, so 
$\Delta(f_i(\bm{x}),P_i(\bm{x})) = \frac{1}{2}\frac{\Delta(r_{i,\bm{x}},C_{i,\bm{x}})}{(n-d+i-\beta \sqrt{n}-e)/2}$.

\item
The decoding succeeds, but is the wrong codeword, whose leading coefficient disagrees with that of the correct codeword.  
In this case, $D_{i,\bm{x}} \neq C_{i,\bm{x}}$, so 

\begin{eqnarray*}
\Delta(f_i(\bm{x}),P_i(\bm{x})) &=& 1-\frac{1}{2}\frac{\Delta(r_{i,\bm{x}},D_{i,\bm{x}})}{(n-d+i-\beta \sqrt{n}-e)/2}\\
&\leq& 1-\frac{(n-d+i)-\Delta(r_{i,\bm{x}},C_{i,\bm{x}})}{(n-d+i-\beta \sqrt{n}-e)}\\
&\leq& 1-\frac{(n-d+i-\beta \sqrt{n}-e)-\Delta(r_{i,\bm{x}},C_{i,\bm{x}})}{(n-d+i-\beta \sqrt{n}-e)}\\
&\leq& \frac{\Delta(r_{i,\bm{x}},C_{i,\bm{x}})}{(n-d+i-\beta \sqrt{n}-e)}.
\end{eqnarray*}

\item
The decoding succeeds, but is the wrong codeword, whose leading coefficient matches that of the correct codeword.  
As in the previous case, $D_{i,\bm{x}} \neq C_{i,\bm{x}}$, and we have:

\begin{eqnarray*}
\Delta(f_i(\bm{x}),P_i(\bm{x})) &=& \frac{1}{2}\frac{\Delta(r_{i,\bm{x}},D_{i,\bm{x}})}{(n-d+i-\beta \sqrt{n}-e)/2} \\
&\leq& 1-\frac{1}{2}\frac{\Delta(r_{i,\bm{x}},D_{i,\bm{x}})}{(n-d+i-\beta \sqrt{n}-e)/2} \\
&\leq& \frac{\Delta(r_{i,\bm{x}},C_{i,\bm{x}})}{(n-d+i-\beta \sqrt{n}-e)}.
\end{eqnarray*}

\end{enumerate}

Putting it all together, we have:

\begin{eqnarray*}
\Delta(f_i,P_i) &\leq& \summ_{\bm{x}\in S^m} {\frac{\Delta(r_{i,\bm{x}},C_{i,\bm{x}})}{n-d+i-\beta \sqrt{n}-e}}\\
&=& \frac{\Delta(r_i,C_i)}{n-d+i-\beta \sqrt{n}-e} \\
&=& \frac{\Delta(r,C)}{n-d+i-\beta \sqrt{n}-e} \\
&\leq& \frac{\frac{n^{m+1}}{2}\left(1-\frac{d + (m+1)\beta\sqrt{n} + e}{n}\right)}{n-d+i-\beta \sqrt{n}-e}\\
&=& \frac{n^m}{2} \frac{n-d-(m+1)\beta\sqrt{n}-e}{n-d+i-\beta \sqrt{n}-e}\\
&\leq& \frac{n^m}{2} \frac{n-d-m\beta\sqrt{n}}{n-d+i}\\
&=& \frac{n^m}{2} \left(1-\frac{i+m\beta\sqrt{n}}{n-d+i}\right).
\end{eqnarray*}

\medskip\noindent\textbf{Analysis of Runtime}

The algorithm can be divided into two parts:

\begin{enumerate}
\item Constructing the $f_i$, $i=0,\ldots, d$.
\item Decoding the $f_i$ to get the $P_i$, $i=0,\ldots, d$.
\end{enumerate}

The dominant contribution to the runtime when constructing $f_i$ comes from all the Reed-Solomon decodings with 
uncertainties we have to do to get the $D_{i,\bm{x}}(Y)$.  For every $e+1$ iterations, we have to decode each row $x\in S^m$ again.  
The total number of such decodings is given by $\frac{n}{e+1}\cdot n^m = \frac{n^{m+1}}{e+1}$.  
Since each Reed-Solomon decoding with uncertainty can be done in $O(n\polylog n)$ time via the FAST-RS-DECODER, we have that the 
runtime of this part of the algorithm is $O\left(\frac{n^{m+2}}{e+1}\polylog n\right)$.

To understand the runtime of the second part of the algorithm, we will compute the runtime of decoding $f_i$ for some fixed $i$.  
Note that decoding $f_i$ is a Reed-Muller decoding with uncertainties problem with $m$ variables.  
So we will write the decoding radius $\frac{n^m}{2} \left(1-\frac{i+m\beta\sqrt{n}}{n-d+i}\right)$ in the form 
$\frac{n^m}{2}\left(1-\frac{i + m\beta\sqrt{n} + e_i}{n}\right)$ and apply the induction hypothesis to get a 
$O\left(\frac{n^{m+1}}{e_i+1}\cdot \polylog n\right)$ runtime.  We now need to compute $e_i$:

\begin{eqnarray*}
e_i &=& n\frac{i+m\beta\sqrt{n}}{n-d+i} - (i + m\beta\sqrt{n})\\
&=& (i + m\beta\sqrt{n}) \left(\frac{n}{n-d+i}-1 \right)\\
&=& \frac{(i + m\beta\sqrt{n})(d-i)}{n-d+i}.
\end{eqnarray*}

The runtime for all $d+1$ iterations from $i=0,\ldots,d$ is then 
$$O\left(\summ_{i=0}^d \frac{1}{e_i+1} \cdot n^{m+1} \polylog n\right).$$

It remains to bound $\summ_{i=0}^d \frac{1}{e_i+1}$ from above:

\begin{eqnarray*}
&& \summ_{i=0}^d \frac{1}{e_i+1}\\
&\leq& \summ_{i=0}^d \min\left(1,\frac{1}{e_i}\right) \\
&\leq& 4+\summ_{i=2}^{d-2} \frac{1}{e_i} \\
&\leq& 4+\int_{1}^{d-1} \frac{n-d+t}{(t + m\beta\sqrt{n})(d-t)} dt.
\end{eqnarray*}

The last inequality is a simple Riemann sum bound using the fact that the function $\frac{n-d+t}{(t + m\beta\sqrt{n})(d-t)}$ 
decreases then increases continuously on $[1,d-1]$.  Computing the integral is a straightforward partial fraction decomposition:

\begin{eqnarray*}
&& \frac{n-d+t}{(t + m\beta\sqrt{n})(d-t)} \\
&=& \frac{n}{(t + m\beta\sqrt{n})(d-t)} - \frac{1}{t + m\beta\sqrt{n}} \\
&=& \frac{n}{d+m\beta\sqrt{n}} \left(\frac{1}{t+m\beta\sqrt{n}} + \frac{1}{d-t} \right) - \frac{1}{t+m\beta\sqrt{n}} \\
&\leq& \frac{1}{\alpha} \left(\frac{1}{t+m\beta\sqrt{n}} + \frac{1}{d-t} \right) - \frac{1}{t+m\beta\sqrt{n}} \\
&=& \left(\frac{1}{\alpha}-1\right) \frac{1}{t+m\beta\sqrt{n}} + \frac{1}{\alpha} \cdot \frac{1}{d-t}
\end{eqnarray*}

So we have:

\begin{eqnarray*}
&& \int_{1}^{d-1} \frac{n-d+t}{(t + m\beta\sqrt{n})(d-t)} dt \\
&\leq& \int_{1}^{d-1} \left[\left(\frac{1}{\alpha}-1\right) \frac{1}{t+m\beta\sqrt{n}} + \frac{1}{\alpha} \cdot \frac{1}{d-t} \right] dt \\
&\leq& O\left(\left(\frac{1}{\alpha}-1\right) \log n + \frac{1}{\alpha} \log n \right) \\
&=& O\left(\left(\frac{2}{\alpha}-1\right) \log n \right) \\
&=& O(\log n).
\end{eqnarray*}

So the runtime for all $d+1$ iterations is:
$$O\left((4+O(\log n)) \cdot n^{m+1} \polylog n\right) = O(n^{m+1} \polylog n).$$

This means the runtime for both parts of the algorithm is just $O\left(\frac{n^{m+2}}{e+1}\polylog n\right)$.

\medskip\noindent\textbf{Base Case}

The algorithm for $m=2$ is almost identical to that for general $m$, except that we decode $f_i(X)$ to a degree $i$ polynomial 
within the larger radius $\frac{n}{2}\left(1-\frac{i+\beta\sqrt{n}}{n}\right)$ to get $Q_i(\bm{X})$.  Note that this radius is 
still less than half the minimum distance of the Reed-Solomon code of degree $i$.  The correctness of the algorithm follows 
from the fact that $P_i$ is still the unique polynomial within distance $\frac{n}{2}\left(1-\frac{i+\beta\sqrt{n}}{n}\right)$ 
of $f_i$.  

We can again analyze the runtime of the two parts of the algorithm.  The runtime for finding the $f_i$ follows the same analysis 
as before and is $O(\frac{n^3}{e+1} \polylog n)$.  For decoding the $f_i$, we simply call the FAST-RS-DECODER for $d+1$ different values of $i$.  
This has a runtime of $O(dn \polylog n) \leq O(n^2 \polylog n)$.  So we get a total runtime of $O(\frac{n^3}{e+1} \polylog n)$.

\end{proof}

The algorithm for general Reed-Muller decoding follows the same strategy as the algorithm for Reed-Muller decoding with uncertainties to a 
radius less than half the minimum distance.  Recall that to get the $f_i$ in the algorithm, we only needed to Reed-Solomon decode to a radius 
significantly less than half 
the minimum distance.  We then saved on the number of Reed-Solomon decodings by instead decoding to half the minimum distance and reusing 
that decoding for many iterations.  We now want to Reed-Muller decode to near half the minimum distance.  Using the same algorithm doesn't 
save on enough Reed-Solomon decodings to achieve near linear time.  However, when there are no uncertainties in the original received word, 
we can list decode efficiently to a radius significantly larger than half the minimum distance.  We then use the lists for many iterations 
to generate the $f_i$ before list decoding again.

\begin{proof}[Proof of Theorem~\ref{GenRMD}]
In the case where the number of variables is $2$, we are in the setting of decoding bivariate Reed-Muller codes
to near half the minimum distance, which can be done in near-linear time by Theorem~\ref{BRMD}.  Assume now that 
$m\geq 2$ and that we have a Reed-Muller code in $m+1$ variables.

\begin{algorithm}[H]
\caption{Decoding Reed Muller\label{alg5}}

\begin{algorithmic}[1]
\State Input: $r:S^{m+1}\to\F$.
\State Let $c = ((1-\alpha)^2/8)$.
\For{$j=0,1,\ldots, \frac{d}{2cn}$}
 \State Let $t_j=\frac{n-d+j\cdot 2cn}{2} + cn$.
 \State Define $r_{j\cdot 2cn}:S^{m}\times S\to\F$ by 
 $$r_{j\cdot 2cn}(\bm{X},Y) = r(\bm{X},Y) - \summ_{i=0}^{j\cdot 2cn-1} Q_i(\bm{X})Y^{d-i}.$$
 \For{$\bm{x}\in S^{m}$}
  \State Define $r_{j\cdot 2cn,\bm{x}}:S\to\F$ by 
	$$r_{j\cdot 2cn,\bm{x}}(Y)= r_{j\cdot 2cn}(\bm{x},Y).$$

	\State Define $L_{j,0,\bm{x}} = \text{RS-LIST-DECODER}(r_{j\cdot 2cn,\bm{x}}(Y),d-j\cdot 2cn,t_j)$.
 \EndFor
	
 \For{$k=0,1,\ldots,2cn-1$}
  \For{$\bm{x}\in S^{m}$} 
	 \State Define $(G_{x}(Y),\delta_x)\in L_{j,k,x}$ to be the unique codeword (if any) with 
	 $$\delta_x < \frac{n-d+j\cdot 2cn + k}{2}$$
	 \State $\sigma_{x}\gets \text{Coeff}_{Y^{d-j\cdot 2cn - k}}(G_{x})$.
	\EndFor
  
	 \State Define the weighted function $f_{j\cdot 2cn+k}:S^m\to\F\times[0,1]$ by
	 $$f_{j\cdot 2cn+k}(\bm{x})=\left(\sigma_{\bm{x}},\min\left\{1,\frac{\delta_{\bm{x}}}{(n-d+j\cdot 2cn + k)/2}\right\}\right).$$
	\State Define $Q_{j\cdot 2cn+k}:S^m\to \F$ by
	 $$Q_{j\cdot 2cn+k}(\bm{X})= \text{RM-UNC-DECODER}\left(f_{j\cdot 2cn+k}(\bm{X}),j\cdot 2cn+k,\frac{n^{m-1}}{2}\left(1-\frac{j\cdot 2cn+k+(m-1)\beta\sqrt{n}}{n-d+j\cdot 2cn+k}\right)\right).$$
	
	\For{$\bm{x}\in S^{m}$}
	 \State $L_{j,k+1,\bm{x}}\gets \{(C-Q_{j\cdot 2cn+k}(\bm{x})Y^{d-j\cdot 2cn - k},\delta_{C,\bm{x}})| (C,\delta_{C,\bm{x}})\in L_{j,k,\bm{x}}, \text{Coeff}_{Y^{d-j\cdot 2cn - k}}(C)=Q_{j\cdot 2cn+k}(\bm{x})$\}.
	\EndFor
 \EndFor
\EndFor
\State Output: $\summ_{i=0}^{d} Q_i(\bm{X})Y^{d-i}$.

\end{algorithmic}
\end{algorithm}

The decoding algorithm for a $m+1$-variate Reed-Muller code is as follows: 
In the $i$-th iteration, list decode row $r_{i,\bm{x}}$, $\bm{x}\in S^m$ to obtain a list $L_{i,\bm{x}}$ of all degree 
$\leq d-i$ polynomials within radius $\frac{1}{2}(n-d+i + cn)$ along with their distances from $r_{i,\bm{x}}$, 
where $c = \frac{(1-\alpha)^2}{8}$.  
Search the list to get the degree $\leq d-i$ polynomial within distance $\frac{1}{2}(n-d+i)$ from $r_{i,\bm{x}}$, 
call it $D_{i,\bm{x}}(Y)$.  We use the lists for $cn$ iterations before list decoding again.
Construct function $f_i:S^m\to \F\times [0,1]$ of (leading coefficient, uncertainty) 
$= \left(\text{Coeff}_{Y^{d-i}}(D_{i,\bm{x}}),\frac{\Delta(r_{i,\bm{x}},D_{i,\bm{x}})}{(n-d+i)/2}\right)$.  
Decode $f_i(\bm{X})$ to a degree $i$ polynomial within radius $\frac{n^{m}}{2}\left(1-\frac{i+m\beta\sqrt{n}}{n-d+i}\right)$ 
to get $Q_i(\bm{X})$.

\medskip\noindent\textbf{Proof of Correctness}

As before, we want to show that $Q_i(\bm{X}) = P_i(\bm{X})$.  It is enough to show 
$$\Delta(f_i,P_i) < \frac{n^{m}}{2}\left(1-\frac{i+m\beta\sqrt{n}}{n-d+i}\right).$$ 

We can use a similar analysis of $\Delta(f_i,P_i)$ to the one in Theorem~\ref{UncRMD} to get to the following step:
$$\Delta(f_i,P_i) \leq \frac{\Delta(r,C)}{n-d+i}.$$

So we have:

\begin{eqnarray*}
\Delta(f_i,P_i) &\leq& \frac{\frac{n^{m+1}}{2}\left(1-\frac{d + m\beta\sqrt{n}}{n}\right)}{n-d+i}\\
&=& \frac{n^{m}}{2} \frac{n-d-m\beta\sqrt{n}}{n-d+i}\\
&=& \frac{n^{m}}{2}\left(1-\frac{i+m\beta\sqrt{n}}{n-d+i}\right).
\end{eqnarray*}

\medskip\noindent\textbf{Analysis of Runtime}

Decoding the $f_i$ over the $d+1$ values of $i$ can be done in $O(n^{m+1} \polylog n)$ following the same runtime analysis from 
Theorem~\ref{UncRMD}.  For constructing the $f_i$, we do $O(n^m)$ Reed-Solomon list decodings taking $O(n \polylog n)$ 
time each.  Within any given list, we need to compute uncertainties for each element of the list.  This also takes 
$O(n \polylog n)$ time for each list.  Finally, we update the lists at each iteration by identifying the elements with 
the correct leading coefficient and taking away their leading terms.  Since the list size is constant, and there are $O(n^{m})$ 
lists to update in each iteration, the updating takes $O(n^{m}d) = O(n^{m+1})$ over $d+1$ iterations.  Hence the total runtime is 
$O(n^{m+1} \polylog n)$ as desired.

\end{proof}

\section{Open Problems}
We conclude with some open problems.

\begin{enumerate}
\item The problem of list-decoding multivariate 
polynomial codes up to the Johnson radius is a very interesting open problem left open by our work. Generalizing our approach seems to require progress on another very interesting open problem, that of list-decoding Reed-Solomon concatenated codes. See~\cite{Guruswami-Sudan-soft-info} for the state of the art on this problem.

\item It would be interesting to understand the relationship between our algorithms and the $m+1$-variate interpolation-based list-decoding algorithm of Sudan~\cite{Sudan}. Their decoding radii are incomparable, and perhaps there is some insight into the polynomial method, which is known to face some difficulties in $> 2$ dimensions, that can be gained here.

\item It would be interesting to see if one can decode multiplicity codes~\cite{KSY} on arbitrary product sets upto half their minimum distance. Here too, we know algorithms that decode upto the minimum distance only in the case when $S$ is very algebraically special (from~\cite{mult-part2}), or if the degree $d$ is very small compared to $|S|$ (via an $m+1$-variate interpolation algorithm, similar to~\cite{Sudan}).

\end{enumerate}

\section*{Acknowledgments}

We are grateful to Madhu Sudan for introducing this problem to us many
years ago.

\titleformat{\section}{\Large\bfseries}{\appendixname~\thesection :}{0.5em}{}
\begin{appendices}

\appendix
\section{Near-Linear Time Soft Decoding of Reed-Solomon Codes}

In this section, we present a near-linear time algorithm to soft decode Reed-Solomon codes to almost 
half the minimum distance.  This result can be used to achieve near-linear time decoding of Reed-Muller 
codes to almost half the minimum distance.

\begin{lemma}
\label{SoftRS}
Let $\F$ be a finite field and let $S\subseteq \F$ be a nonempty subset of size $|S|=n$.  
There is a randomized algorithm $\text{FAST-RS-DECODER}(r,d)$ that given a received word 
with uncertainties $r:S\to\F\times[0,1]$, finds the unique polynomial (if it exists) 
$C\in\F[X]$ satisfying $\deg(C)\leq d$ and $\Delta(r,C) < \frac{n-d-\sqrt{n}}{2}$ with probability 
$3/4$ in time $O(n \polylog(n))$. 
\end{lemma}

\begin{proof}
The near-linear time algorithm for $\text{FAST-RS-DECODER}(r,d)$ is based on 
Forney's generalized minimum distance decoding of concatenated codes.  

Given a received word $r:S\to \F\times [0,1]$, suppose there is a polynomial $f$ 
of degree at most $d$ such that $\Delta(f,r) < \frac{n-d-\sqrt{n}}{2}$.  Let 
$S = \{\alpha_1,\alpha_2,\ldots,\alpha_n\}$, and write $r(\alpha_i) = (\beta_i,u_i), i\in [n]$.  
We may view $r$ as a set of $n$ points $(\alpha_i,\beta_i)$ with uncertainties $u_i$.  
The general idea of the algorithm is to erase the $i$-th point with probability $u_i$, 
and perform errors and erasures decoding of the resulting Reed-Solomon code.  We denote 
the errors and erasures Reed-Solomon decoder by $\text{EE-DECODER}(r',d)$, which takes a 
received word $r':S\to \F\times[0,1]\cup \{?\}$ and a degree $d$ and returns the polynomial 
of degree at most $d$ that is within $\frac{n-d}{2}$ of $r'$.

\begin{algorithm}[H]
\caption{Fast Reed-Solomon Decoding with Uncertainties\label{alg3}}

\begin{algorithmic}[1]
\State Input: $r:S\to\F\times[0,1]$.
\For{$i = 1,2,\ldots, n$}
 \State $p_i\gets \text{RANDOM}([0,1])$.
 \State Define $r':S\to\F\times[0,1]\cup\{?\}$ by 
  $$r'(\alpha_i) = \begin{cases} \beta_i & \text{if } p_i \leq u_i \\ ? & \text{if } p_i > u_i \end{cases} .$$
\EndFor
\State $g\gets\text{EE-DECODER}(r',d)$.
\State Output: $g$.
\end{algorithmic}
\end{algorithm}

We say that a point is an $\emph{erasure}$ if it is erased by the algorithm.  We say that 
a point $(\alpha_i,\beta_i)$ is an $\emph{error}$ if $(\alpha,\beta)$ is not an erasure and 
$f(\alpha_i)\neq \beta_i$.  Let $E$ be the number of errors, and let $F$ be the number of 
erasures.  As the resulting $n-F$ points form a Reed-Solomon code of block length $n-F$ and degree $d$, 
the algorithm outputs $f$ as long as 

$$2E+F < n-d.$$

We will use Chebyshev's inequality to show that $2E+F < n-d$ with probability at least $\frac{3}{4}$.  
To help us compute the expectation and variance of $2E+F$, we write $E$ and $F$ as a sum of indicator 
random variables.  Let $A = \{i\in[n]| f(\alpha_i)= \beta_i\}$ be the set of agreeing indices, and 
let $D = \{i\in[n]| f(\alpha_i)\neq \beta_i\}$ be the set of disagreeing indices.  
Let $T = \{i\in[n]| (\alpha_i,\beta_i) \text{ is erased}\}$ be the set of erasure indices.

Then we can write
\begin{eqnarray*}
E &=& \summ_{i\in D}{1_{i\notin T}}\\
F &=& \summ_{i\in [n]}{1_{i\in T}}.
\end{eqnarray*}

We then can show $\E[2E+F]$ is less than $n-d$ by a significant amount $\sqrt{n}$: 
\begin{eqnarray*}
\E[2E+F] &=& 2\summ_{i\in D}{(1-u_i)} + \summ_{i\in [n]}{u_i}\\
&=& 2\summ_{i\in D}{(1-u_i)} + \summ_{i\in D}{u_i} + \summ_{i\in A}{u_i}\\
&=& 2\left(\summ_{i\in D}{\left(1-\frac{u_i}{2}\right)} + \summ_{i\in A}{\frac{u_i}{2}}\right)\\
&=& 2\Delta(f,r)\\
&<& n-d-\sqrt{n}.
\end{eqnarray*}

Finally, we show that $\Var(2E+F)$ is small:
\begin{eqnarray*}
&&\Var(2E+F)\\
&=& 4\Var(E) + 4\Cov(E,F) + \Var(F)\\
&=& 4\summ_{i\in D}{u_i(1-u_i)} + 4\left(\E\left[\summ_{i\in D}{\summ_{j\in[n]}{1_{i\notin T \cap j\in T}}}\right] 
- \summ_{i\in D}{(1-u_i)}\summ_{j\in [n]}{u_j}\right) + \summ_{i\in[n]}{u_i(1-u_i)}\\
&=& 4\summ_{i\in D}{u_i(1-u_i)} + 4\left(\E\left[\summ_{i\in D}{\summ_{j\neq i}{(1-u_i)u_j}}\right] 
- \summ_{i\in D}\summ_{j\in [n]}{(1-u_i)u_j}\right) + \summ_{i\in[n]}{u_i(1-u_i)}\\
&=& 4\summ_{i\in D}{u_i(1-u_i)} - 4\summ_{i\in D}{u_i(1-u_i)} + \summ_{i\in[n]}{u_i(1-u_i)}\\
&=& \summ_{i\in[n]}{u_i(1-u_i)}\\
&\leq& \frac{n}{4}.
\end{eqnarray*}

By Chebyshev's inequality, $\Pr(2E+F\geq n-d) \leq \frac{1}{4}$.  Hence we have $\Pr(2E+F < n-d) \geq \frac{3}{4}$.  
That is, with probability at least $\frac{3}{4}$, the algorithm outputs $f$.  

We now analyze the runtime of our fast Reed-Solomon decoder.  
The erasures can be done in $O(n)$ time.  
Also, as the EE-DECODER is essentially a Reed-Solomon decoder to half the minimum distance, it runs in time 
$O(n \polylog n)$~\cite{Berl,WB}.  This gives a total runtime of $O(n \polylog n)$.

\end{proof}

Note that by repeating the algorithm $\Omega(\log n)$ times, 
we find the unique codeword in $O(n \polylog n)$ time with probability $1-1/n^{\Omega(1)}$ (the $\Omega(1)$ can be chosen to 
be an arbitrary constant).

\end{appendices}



\end{document}